\documentclass[conference]{IEEEtran}
\ifCLASSINFOpdf
\else
\fi
\usepackage{kerger}

\begin{document}
%
\title{Classical and Quantum Distributed Algorithms for the Survivable Network Design Problem}



\author[1,2,3]{Phillip A. Kerger}

\author[2,3,4]{David E. Bernal Neira}

\author[2,3]{Zoe Gonzalez Izquierdo}

\author[3]{Eleanor G. Rieffel}

\affil[1]{Department of Applied Mathematics and Statistics, Johns Hopkins University, Baltimore MD}

\affil[2]{Research Institute of Advanced Computer Science, USRA, Mountain View CA}

\affil[3]{Quantum Artificial Intelligence Laboratory, NASA Ames Research Center, Mountain View CA}

\affil[4]{Department of Chemical Engineering, Purdue University, West Lafayette IN}

\maketitle

\begin{abstract}
We investigate distributed classical and quantum approaches for the survivable network design problem (SNDP), sometimes called the generalized Steiner problem. These problems generalize many complex graph problems of interest, such as the traveling salesperson problem, the Steiner tree problem, and the k-connected network problem. To our knowledge, no classical or quantum algorithms for the SNDP have been formulated in the distributed settings we consider. We describe algorithms that are heuristics for the general problem but give concrete approximation bounds under specific parameterizations of the SNDP, which in particular hold for the three aforementioned problems that SNDP generalizes. We use a classical, centralized algorithmic framework first studied in \cite{GoemansBertsimas1993SurvivableNetworks} and provide a distributed implementation thereof. Notably, we obtain asymptotic quantum speedups by leveraging quantum shortest path computations in this framework, generalizing recent work of \cite{Kerger2023_mind_the_O-tilde}. These results raise the question of whether there is a separation between the classical and quantum models for application-scale instances of the problems considered.
\end{abstract}


%
\IEEEpeerreviewmaketitle

\section{Introduction}
The survivable network design problem (SNDP), also referred to as the generalized Steiner problem, is a remarkably general NP-hard problem that encompasses many graph optimization problems of interest as special cases, such as the traveling salesperson problem (TSP henceforth), the $k$-connected network problem, the Steiner tree problem, and thus also the $s$-$t$ shortest path and minimum spanning tree (MST) problems. In the SNDP, each node is prescribed some \textit{connectivity number}, and the goal is to find the minimum cost multigraph such that for all $k$,  all nodes with connectivity number $\geq k$ have at least $k$ edge-disjoint paths between them. Thus, in a solution to an SNDP, if any $k$ edges are deleted, nodes with connectivity number greater than $k$ remain connected in the network. Hence, the SNDP serves to model situations in which particular nodes remaining connected is critical, with that connectivity guaranteed to withstand, or \textit{survive}, some failures in the network.   

Distributed computing naturally models some network of processors that each can compute and communicate with one another, so it is natural to consider the SNDP in a distributed model of computation, as we do in this paper.
The classical CONGEST-CLIQUE Model (cCCM henceforth) in distributed computing has been studied as a model central to the field, e.g., \cite{TowardsCCM_Complexity_KorhonenSumoela2018, Saikia_Karmakar2019_SteinerTree_CONGESTCLIQUE, fischer2021_DMST, Lenzen2012_OptimalRoutingSortingInCongestClique, Dolev_lenzen_2012TriTA, MST_in_O1_CCM_Nowicki2019}, but to the best of the authors' knowledge, no algorithms for the SNDP have appeared in the CONGEST-CLIQUE or the closely related CONGEST model. Here we describe algorithms in those models, and also in the quantum CONGEST-CLIQUE model. The quantum algorithm exhibits asymptotic quantum advantage in terms of rounds of communication. 

In the cCCM, each node in the network has some input information, but no node has global information. The processors at the network nodes must work together to solve a problem 
under significant communication limitations, described in detail in \S \ref{sec: background}. 
For example, the cCCM can model a network of aircraft or spacecraft, satellites, and control stations, separated by large distances, that have severely limited communication between them.
Solving classical graph problems such as minimum spanning trees, shortest paths, matchings, Steiner trees, and more on such networks may provide helpful information about the network. For example, if each edge weight represents the cost of sending an item between its incident nodes, then the shortest (weighted) path from a node $u$ to another node $v$ gives the minimum cost path to send an item from $u$ to $v$. 
The cCCM aims to model networks in which the available communication bandwidth between nodes primarily limits the process of solving such problems. 
A good introduction to cCCM, and models with limited communication bandwidth more generally, 
can be found in \cite{DistGraphAlgos_LectureNotes_Ghaffari}.
The quantum version of the cCCM, the quantum CONGEST-CLIQUE Model (qCCM), as well as the quantum CONGEST model, have been the subject of recent research \cite{IzumiLeGallMag2019_APSP_QuantumDist, CensorHillel2022_QuantDistCliqueDetect, vanApeldoorn2022_DistQuantQueriesCONGEST, Elkin2014_NoQuantumSpeedups, Kerger2023_mind_the_O-tilde} seeking to understand how quantum communication may help in limited communication bandwidth distributed computing frameworks. 

\subsection{Contributions and Relation to Existing Work} 
Algorithms for the SNDP have been thoroughly studied in the classical, centralized setting (e.g., \cite{Goemans_williamson_1992_constrained_forest_problems}, \cite{Jain1998_2_approx_general_steiner}, \cite{Williamson1993_primal_dual_approx_for_general_steiner}, with the work of \cite{GoemansBertsimas1993SurvivableNetworks} of particular relevance providing algorithms that are, in general, heuristics, but give provable approximation bounds under various specific settings of the connectivity numbers. 
The primary contribution of this paper is to provide distributed algorithms for the SNDP in both the classical and quantum CONGEST-CLIQUE models. These algorithms can be viewed as distributed implementations of the \textit{tree heuristic} in \cite{GoemansBertsimas1993SurvivableNetworks}. 
The algorithm is thus also a heuristic in general, but inherits the same approximation ratios for special cases of the SNDP, summarized in table \ref{table: sndp guarantees}. In particular, for the special cases of Steiner trees, TSP, and the $k$-connected subnetwork problem, the approximation ratio is less than $2$. 
The algorithms use shortest path computations and routing tables as subroutines, 
taking the approach of \cite{GoemansBertsimas1993SurvivableNetworks} to use spanning trees on the complete distance graph of the input graph. 
Using this larger algorithmic framework and building a distributed implementation thereof, the bottleneck in the algorithms is the computation of shortest path distance and routing tables (which store the information about the shortest paths in a distributed manner). The best complexities 
are achieved by using the respectively fastest APSP algorithms in the classical and quantum CONGEST-CLIQUE models. In the classical model, this results in a complexity of $\tilde{O}(n^{1/3})$ using the algorithms of \cite{CensorHillel2016_AlgebraicMethodsCCM_fast_APSP} for the APSP subroutine, and thus also for the whole algorithm, while in the quantum model our method inherits a speedup by using the approach of \cite{IzumiLeGallMag2019_APSP_QuantumDist} with routing tables as in \cite{Kerger2023_mind_the_O-tilde}. While the practicality of these two algorithms may be limited (see \cite{Kerger2023_mind_the_O-tilde}) because the asymptotic advantage only kicks in a impractically large graph sizes, the formulation of asymptotically faster algorithms is the starting point from which to explore the possibility of finding practical faster algorithms. In particular, with the great generality that the SNDP provides, faster quantum algorithms for this problem would be valuable.  

Furthermore, the \textit{improved} tree heuristic of \cite{GoemansBertsimas1993SurvivableNetworks} depends on computing minimum weight maximal matchings and may be viewed as a generalization of Christofides' algorithm for the TSP \cite{Christofides1976WorstCaseAO} to the SNDP. Hence, it would be of future interest to formulate faster algorithms for the min weight max matching problem in the distributed setting. That would allow for improvements to our algorithm, and also have its own merit especially in other directions, including to the decoding of quantum error correction codes (see, e.g., \cite{Ambainis2005QuantumMatchingAndNetworkFlows} for seminal work, and \cite{Wu2023Fusion_blossom} for a more recent state-of-the-art algorithm). 
A family of problems studied in \cite{Goemans_williamson_1992_constrained_forest_problems} (in a centralized setting) called Constrained Forest Problems exhibits a broad generality, going beyond SNDP to generalize the minimum-weight-perfect-matching problem, along with Steiner trees (thus MSTs and $s$-$t$ shortest paths), TSP, and the $T$-join problem. A distributed implementation of this algorithm, however, seems much more difficult, but would be valuable work to pursue in the future. The challenge in distributing this algorithm is that a forest is maintained through iterations of the algorithm, each of which adds one edge to the forest. Adding such an edge effectively causes the other edges to reweight, affecting which edges will be added. Hence, one cannot add edges to the forest simultaneously while maintaining the same guarantees, which stifles an efficient distributed implementation.  


\section{Background}\label{sec: background}
Let us first introduce the notation to be used throughout this paper.
\subsection{Notation}
For an integer-weighted graph $G = (V, E, W)$,
we will denote $n := |V|, m := |E|,$ and $W_e$ the weight of an edge $e \in E$ throughout the paper. These will always refer to the relevant objects of the problem input graph, unless otherwise specified.
Let $\delta(v) \subset V$ be the set of edges incident on node $v$, and $\mathcal{N}_{G}(u) := \{v: uv \in E\}$ the neighborhood of $u\in G$.
Denote by $d_G(u,v)$ the shortest-path distance in $G$ from $u$ to $v$.
All logarithms are in base 2 unless specified.

We begin by defining the survivable network problem in its full generality and specifying how it generalizes other notable problems.
\begin{definition}[Survivable Network Design Problem]$\phantom{x}$ \\ \smallskip
Given a weighted, undirected graph $G=(V, E, W)$ and a {\em connectivity number} $r_v$ for each node $v\in V$, find a minimum weight subgraph such that any pair of nodes $(u,v)$ is connected by $min\{r_u, r_v\}$ edge-disjoint paths, allowing the subgraph to be a multigraph, i.e. allowing multiple copies of edges of $G$. The solution subgraph is called the minimum weight survivable network.
\end{definition}
The practical interpretation of the SNDP is that in a real network, for a solution to an SNDP, if $k$ links between nodes fail, all nodes with connectivity number $\geq k$ remain connected in the subgraph, that is, subnetworks of nodes with connectivity greater than $k$ can tolerate $k$ edge failures and remain connected. For simplicity of terminology, we will simply use the term subgraph to include a subgraph that is a also multigraph (i.e., a subgraph of $G$ but allowing for edges to appear multiple times). The solutions we build for the SNDP will be precisely such subgraphs.  \\
This generalizes the following problems as follows: 
\begin{itemize}
    \item[i)] Steiner tree with Steiner nodes $Z$: Set $r_z=1$ for all $z\in Z$.  
    \item[ii)] TSP, under the condition that the graph satisfies the triangle inequality: Set $r_v =2$ for all $v\in V$.
    \item[iii)] $k$-connected network problem: Set $r_v = k$ for all $v\in V$.
    \item[iv)] Minimum spanning tree: Set $r_v = 1$ for all $v\in V$. 
    \item[v)] $s$-$t$ shortest paths: Set $r_s = 1, r_t = 1$. 
\end{itemize}

Our algorithms will only approximately solve the SNDP. To be precise, we mean the following:
\begin{definition}
    Given a weighted graph $G = (V, E, W)$, let $S^*$ be the solution to an instance of the SNDP, with edge $e\in E$ appearing $k_{e}$ times in $S^*$. A subgraph (and multigraph) $\hat{S}$ with edge $e$ appearing $\hat{k}_{e}$ times that satisfies the connectivity requirements given by the SNDP instance and has 
    $$
    \sum_{e\in E}\hat{k}_e \cdot W_e \leq \gamma \sum k_e \cdot W_e
    $$
    is a $\gamma$-approximate solution to the SNDP.
\end{definition}

The algorithms presented in this paper are largely dependent on shortest path computations. However, as is often the case, it will not be sufficient to calculate the shortest path distance $d_G (u,v)$ for each pair of nodes $(u,v)\in V\times V$;  information about the shortest \textit{path}, i.e. about the path from $u$ to $v$ achieving that distance, is also needed. The shortest path information will be distributedly stored as so-called \textit{routing tables}.

\begin{definition}\label{def: routing table}
A \textit{routing table} for a node $v$ is a function $R_v: V \to V$ mapping a vertex $u$ to the first node visited in the shortest path going from $v$ to $u$ other than $v$ itself.
\end{definition}

Hence, if nodes have access to their routing tables, they only have knowledge of the first step of each shortest path connecting them to other nodes, but not the full path information. However, for many situations this is enough. Consider, for example, a message being sent from node $v$ to node $u$ that should be sent along the shortest path. Let $v t_1 t_2 ... t_k u$ be the shortest path with intermediate nodes $t_1, ..., t_k$. $v$ only knows from its routing table that $t_1$ is the first node along the path to $u$, i.e. $R_v(u) = t_1$. Node $v$ sends the message first to node $t_1$, which in turn uses its routing table to determine it should pass it on to $t_2$, and so on until $t_k$ passes the message on to $u$.

\subsection{The CONGEST and CONGEST-CLIQUE Models of Distributed Computing}
The standard CONGEST model considers a graph of $n$ processor nodes whose edges also serve as communication channels. Initially, each node knows only its neighbors in the graph and associated edge weights.
Each processor node executes computation locally in rounds and then communicates with its neighbors before performing further local computation.
The congestion limitation restricts this communication, with each node able to send only one message of $O(\log(n))$ classical bits in each round to its neighbors, though the messages to each neighbor may differ. 
Since there are $n$ nodes, assigning them ID labels $1, ..., n$ means the binary encoding size of a label is $\ceil{\log(n)}$ bits -- i.e., each message in CONGEST of $O(\log n)$ bits can contain roughly the amount of information to represent one node ID.
In the cCCM, we separate the communication graph from the problem input graph by allowing all nodes to communicate with each other. However, the same $O(\log(n))$ bits-per-message congestion limitation remains.
Hence, a processor node could send $n-1$ different messages to the other $n-1$ nodes in the graph, with a single node distributing up to $O(n(\log(n))$ bits of information in one round.
Taking advantage of this way of dispersing information to the network is paramount in many efficient CONGEST-CLIQUE algorithms.
The efficiency of algorithms in these models is commonly measured in terms of the \textit{round complexity}, the number of rounds of communication used in an algorithm to solve the problem in question. 

\subsection{Quantum CONGEST-CLIQUE}The quantum CONGEST-CLIQUE model (qCCM) is obtained by the simple modification of allowing nodes to send quantum bits, \emph{qubits}, to one another and enabling the processors in the network to compute with those qubits  (see \S \ref{subsubsec: qc} below for more details). Specifically, in the communication step, each node may send a (possibly different) register of $O(log(n))$ qubits to every other node, and those qubits may be entangled with one another. In the computation part of the rounds, each node can make unlimited use of quantum computing resources, i.e., preparing registers of qubits, applying unitary operators to them, and making measurements. 
Here, we give a formal definition of both the quantum and standard CONGEST-CLIQUE, with the quantum modifications in parentheses. 

\begin{definition}[(Quantum) CONGEST-CLIQUE] \label{def: QCCM}
The (quantum) CONGEST-CLIQUE Model, or (q)CCM, is a model of distributed computation in which an input graph $G = (V, E, W)$ is distributed over a network of $n$ (quantum) processors, where each processor is represented by a node in $V$.
Each node is assigned a unique ID number in $[n]$. Time passes in \textit{rounds}, each of which consists of the following phases:
\begin{enumerate}
    \item[i)] Each node may execute unlimited local (quantum) computation.
    \item[ii)] Each node may send a message consisting of $\mathcal{O}(\log n)$ (quantum) bits to each other node in the network. Each of those messages may be distinct. 
    \item[iii)] Each node receives and saves the other nodes' messages. 
\end{enumerate}
The input graph $G$ is distributed across the nodes as follows: 
Each node knows its own ID number, the ID numbers of its neighbors in $G$, the number of nodes $n$ in $G$, and the weights corresponding to the edges it is incident upon. The output solution to a problem must be given by having each node $v\in V$ return the restriction of the global output to $\mathcal{N}_{G}(v) := \{u: vu \in E\}$, its neighborhood in $G$. In the qCCM, no entanglement is shared across nodes initially.

We remark here that in both the cCCM and qCCM, any problem can be solved naively in $O(n)$ rounds. To see this, simply notice the following. Each node in the network is incident on at most $n$ edges. Sending the information of one edge as a message takes $O(1)$ messages, since $O(log(n))$ is the length of one message, and $log(n)$ bits suffice to encode the ID of a node in the graph. Information about an edge consists of two edges and its weight, so assuming the weight is an integer no more than $2^n$, $3log(n)$ bits suffice to encode the information about one edge. So it takes $O(n)$ rounds, or up to $n$ messages of size $3log(n)$, for a node to broadcast the information about all of its edges to the network. Every node can do this, simultaneously, in $O(n)$ rounds in the model, and so after this process every node in the graph has full information about the graph. Then each node can locally compute the solution to the problem that was posed about the graph. Alternatively, the nodes could all send their edge information to the node with the smallest ID in the network, which would then compute the solution and share it with the network, so that only $O(n^2)$ messages would need to be sent. This poses a particular challenge in formulating meaningful algorithms in the CONGEST-CLIQUE type models, because any algorithm that is not sublinear is no better than this simple approach. 

\begin{remark}
    We state the CONGEST-CLIQUE model as nodes of the graph for which the problem in question needs to be solved being the processors in the network. While the natural interpretation of this setup is that one wants to, for example, find a survivable network or MST etc. over the processors, one can also simply think about this as modeling a distributed computing approach in which one has $n$ processors and is solving a problem about some completely different graph that has $n$ nodes. If one matches up each processor node $v$ to a node $v'$ in the problem graph $G$ of interest and sends the information $\mathcal{N}_G(v')$ to $v$, one can use CONGEST-CLIQUE algorithms to solve the given problem over the processor network. 
\end{remark}

\end{definition}

\subsubsection{Brief Background on Quantum Computing}\label{subsubsec: qc}
While a single classical bit can be in the state 0 or 1, and a system of $n$ classical bits can be represented by the state space $\{0,1\}^n$, the state of a single quantum bit, or \textit{qubit}, can be represented by a unit-length vector in a two-dimensional complex Hilbert space $\cal Q$, with basis vectors commonly written as $\ket{0}$ and $\ket{1}$. A qubit in state $\ket{\psi}$ can be measured with respect to some orthonormal basis $\ket{v_1}, \ket{v_2}$, and the result of that measurement will be $\ket{v_1}$ with probability $\braket{\psi, v_1}^2$, and $\ket{v_2}$ with probability $\braket{\psi, v_2}$. Unit vectors can then represent the state space of a register of $n$ qubits in the $n$-fold tensor product of the single qubit state-spaces $\cal Q \otimes \cal Q \otimes \dots \otimes \cal Q$. Since these must always remain unit vectors (for the measurement probabilities to sum to $1$), operations on qubits are represented by $2^n\cross 2^n$ unitary matrices. Several physical implementations of qubits---such as superconducting qubits, trapped ion qubits, photonic qubits, etc.---currently exist, and are the subject of extensive ongoing research.
For a more thorough introduction to the fundamentals of quantum computing, we refer the reader to \cite{RieffelGentleIntroductionToQuantum}. Celebrated results in quantum computing include Grover's search algorithm, first formulated in \cite{grover1996}, which finds a marked element in an unsorted list of $N$ elements in $O(\sqrt{N})$ time, and Shor's algorithm \cite{Shor1995_og_paper}, which solves the prime factorization problem in polynomial time. On a high level, the intuitive reason quantum computers can give advantages in specific problems is that quantum interference effects can be leveraged to create canceling effects among non-solutions to a problem. The asymptotic quantum speedup for our algorithm stems from taking advantage of the Grover search, which can be used to speed up distributed triangle finding in graphs to allow for faster shortest-path computations. These can then be used for approximation algorithms to the survivable network design problem. 
The distributed version of Grover Search (see e.g. \cite{IzumiLeGall2020_TriangleFinding}), roughly speaking, does the following. If some node $u$ in the network is able to evaluate a function $g: X \rightarrow \{0,1\}$ in $r$ rounds of communication over the network, then distributed Grover search can be used to have $u$ find an element $x\in X$ such that $g(x) = 1$ in $\tilde{O}(r\sqrt{|X|})$ rounds. With no assumptions about $g$, it is straightforward to see that if $g$ can only be evaluated classically, $r|X|$ rounds are needed in the worst case to find $x\in X$ such that $g(x) = 1$. The best one can do is evaluate $g$ for each element in $X$ one-by-one until one finds that $g(x) = 1$. Since each evaluation takes $r$ rounds, $r|X|$ rounds are needed. Evaluating $g$ with inputs in superposition, by sending messages of possibly entangled qubits, can on the other hand be leveraged in the distributed setting to obtain a quadratic speedup in the number of rounds needed with respect to the size of $|X|$. Of course this is unhelpful if $r >> |X|$, but in many cases $|X|$ is much larger, and the speedup becomes meaningful.

\section{A General Distributed Algorithm for SNDP}
We begin by describing a general high-level algorithm for the SNDP first presented in \cite{GoemansBertsimas1993SurvivableNetworks}. Then, we move on to particular implementations in the distributed setting and study their complexities. For each edge $e\in E$, let $x_e$ denote the number of times $x_e$ appears in the solution to the SNDP. 
\begin{mdframed}\namedlabel{alg: tree heuristic}{\texttt{TreeHeuristic}}
\begin{enumerate}
        \item Compute shortest path distances and routing tables for all nodes.
        \item Prepare a list of the connectivity numbers of the nodes, $(c_1, ..., c_k)$, and let it be sorted $c_1 < c_2 < ... < c_k$.
        \item  For each $i = k, k-1, ..., 1$, compute an MST on the subgraph of $\bar{G}$ restricted to nodes $v$ with connectivity requirement $r_v \geq c_i$, and for each edge $uv$ in the resulting MST, increment $x_e$ by $c_i - c_{i-1}$ for every edge $e$ of a shortest path connecting $u$ and $v$.  \\
        \item Implement techniques for finding (local) improvements to the resulting solution. 
\end{enumerate}
\end{mdframed}

The intuition behind this algorithm is relatively straightforward. It is helpful to consider the TSP case, where $r_v = 2$ for all $v\in V$, and disregard step 4. Then, this procedure finds an MST on the graph and includes every edge of the MST twice to obtain the SNDP. Since the weight of the MST is at most the weight of the minimum weight Hamiltonian tour, one immediately obtains a 2-approximation to the TSP. For the Steiner tree case, where $r_v =1$ for Steiner nodes and $r_v = 0$ otherwise, the algorithm finds an MST on the complete distance graph and then constructs the SNDP solution by taking all the shortest paths corresponding to the edges in that MST, which recovers the approach of the approximation algorithm for the Steiner tree problem first introduced by \cite{Kou1981_fast_algo_for_ST_}, and was shown always to give an approximation guarantee of less than 2.

We provide a brief overview of how to implement this meta-algorithm in the distributed setting before formally describing it in the next section. The first step may be approached in various ways, e.g., using the asymptotically fastest means of computing APSP distances in \cite{IzumiLeGallMag2019_APSP_QuantumDist} (modified to handle routing tables in \cite{Kerger2023_mind_the_O-tilde}) or \cite{CensorHillel2016_AlgebraicMethodsCCM_fast_APSP} in the quantum or classical settings, respectively. Using faster approximation algorithms as in \cite{Biswas2020_approximate_shortest_paths} or \cite{Dory2021Constant_round_approx_APSP} is challenging because these typically do not provide routing table information, but approximate distance computations only. The second step can be achieved by having each node broadcast its connectivity number so that every node in the network knows all the other nodes' connectivity numbers and can prepare this sorted list. The MSTs in step 3 can be computed efficiently using, e.g., the algorithm of \cite{MST_in_O1_CCM_Nowicki2019}. A significant difficulty arises in step 3 of going from an edge in $\bar{G}$ to a shortest path connecting $u$ to $v$. This leads us to make use of finding \textit{shortest path forests} (SPF) as done in \cite{Saikia_Karmakar2019_SteinerTree_CONGESTCLIQUE} as a proxy for the complete distance graph for the iterations in step 3); however, the algorithm can still be viewed as an implementation of the above meta-algorithm. 

We provide here a description of the distributed implementation, for the quantum or classical CONGEST-CLIQUE: 
\begin{enumerate}
    \item Compute APSP via some distributed APSP algorithm (such as \cite{IzumiLeGallMag2019_APSP_QuantumDist}, with routing tables in \cite{Kerger2023_mind_the_O-tilde}, or \cite{CensorHillel2016_AlgebraicMethodsCCM_fast_APSP}) to obtain the complete distance graph $\bar{G}$ (or an approximation thereof) of the input graph $G$.
    \item Create a sorted list $\mathcal{L} := ( c_1, ..., c_k)$ of the distinct connectivity numbers (including connectivity number 0), and let $V_j = \{i\in V: r_i \geq c_j\}$ be the set of nodes with connectivity at least $c_j$. 
    \item For $j = k, ..., 1$ (skipping connectivity number 0 if present),\\
    Construct an SPF using $V_j$ as the source nodes for each tree. Modify edge weights as in \cite{Saikia_Karmakar2019_SteinerTree_CONGESTCLIQUE}, and find an MST on this graph. For all edges in the resulting MST, add $c_j - c_{j-1}$ copies of them to the survivable network after pruning excessive edges.   
    \item Search for local improvements, optionally. 
\end{enumerate}

This heuristic effectively does the following in a distributed way. Each distinct connectivity number is treated separately, and the algorithm begins by considering the nodes with the highest connectivity number. Suppose $c_{k}$ and $c_{k-1}$ are the largest and second-to-largest connectivity numbers. A low-cost tree $T_k$ containing the nodes $V_k$ is found, and for each edge in $T_k$, $c_{k} - c_{k-1}$ copies of the edge are added to $S$. Set the connectivity number of nodes $V_k$ to be $c_{k-1}$ so that $c_{k-1}$ is now the maximum connectivity number, and repeat until connectivity number zero is reached.

We make some remarks on the connection to the TSP as a specialization of SNDP. Treating instances of TSP as SNDP and applying the approximation algorithm of \cite{GoemansBertsimas1993SurvivableNetworks}, or the above distributed version we propose, one, in general, does not get Hamiltonian cycles as the ``approximate" solutions, even if they exist. Since finding a Hamiltonian cycle or deciding there is none is itself NP-complete, one should in some sense not expect this of an approximation or polynomial-time (centralized) algorithm (finding a Hamiltonian cycle is not much easier than solving the TSP). So, as our approach is to use a centralized framework based on APSP, which does not need exponential time in a centralized setting, and implement it efficiently in a distributed setting, one cannot hope to obtain Hamiltonian tours. So, we must be satisfied with allowing non-Hamiltonian tours as approximate TSP solutions, and indeed that is what the algorithms presented here obtain.

Further, we remark that for the TSP problem specifically, or other instances where all nodes have the same connectivity number, one can skip the APSP computation in step 1 unless one wishes to make certain local improvements in step 4. In particular, to implement Christofides-type improvements, one would need to compute APSP, as well as minimum weight perfect matchings. Furthermore, suppose one considers a variation of TSP in which not all nodes need to be visited, but some may be used as intermediates. In that case, the simple tree-doubling heuristic fails to give good approximation guarantees, whereas the general SNDP algorithm still applies. 

\subsection{Complexity}
In both the cCCM and qCCM, step 2) and each iteration of step 3) can be implemented in a constant number of rounds. This means that step 1) will usually be the computational bottleneck for the algorithm (as long as there are not too many iterations to be made in step 3), i.e., there are not too many distinct connectivity numbers). 
As a result, the presented SNDP algorithm also permits an asymptotically faster quantum implementation in the quantum CONGEST-CLIQUE model. Namely, the computation of APSP with routing tables in step 1) can be achieved using the algorithm from \cite{IzumiLeGallMag2019_APSP_QuantumDist}, extended to also compute routing tables in \cite{Kerger2023_mind_the_O-tilde}, which pushes the round complexity to $\tilde{O}(n^{1/4})$ instead of $\tilde{O}(n^{1/3})$ compared to using the asymptotically fastest known classical APSP algorithm in the standard CONGEST-CLIQUE model presented in \cite{CensorHillel2016_AlgebraicMethodsCCM_fast_APSP}. As long as there are no more than $O(n^{1/4})$ or $O(n^{1/3})$ distinct connectivity types in the given SNDP, which we consider a reasonable assumption to make. 

\subsection{Approximation Guarantees}
As mentioned previously, our algorithms can be viewed as distributedly realizing the algorithmic framework of \cite{GoemansBertsimas1993SurvivableNetworks} (specifically, the \textit{tree heuristic} algorithm of section 5). Hence, our algorithm inherits the same approximation guarantees for specific parameterizations of the problem, which we state in the following table, 
\begin{table}[htbp]
    \centering
    \caption{Instance of SNDP and Approximation Ratio}
    \begin{tabular}{l c} 
        \toprule
        Instance of SNDP & Approximation Ratio \\
        \midrule
        General & $(2-\frac{2}{|V_1|})\sum_{j=1}^k\frac{c_j - c_{j-1}}{c_j}$ \vspace{1em}\\
        $r_v \in \{0,1,2\} \forall v\in V$ & $2-\frac{2}{|V_1|}$  \\
        or $r_v \in \{0, k\} \forall v\in V$ & \vspace{1em} \\
        $c_1 = 1$ & $(2-\frac{2}{|V_1|})(\sum_{j=1}^k\frac{c_j - c_{j-1}}{c_j}-\frac{1}{2})$ \\
        \bottomrule
    \end{tabular}
    \label{table: sndp guarantees}
\end{table}
where $V_1$ and $c_j$ are as defined in \ref{alg: tree heuristic}, and $r_v$ are the connectivity requirements given in the SNDP. In particular, the special case of Steiner trees, TSP, and $k$-connected subnetwork all fall under the second row of the table, yielding approximation ratio $2-2/|V_1|$ in the worst-case.

\section{Distributed Implementation}
In this section, we describe in more detail the distributed implementation for our algorithms. We defer discussion of distributed APSP computations for step 1) to subsection \ref{sec: distributed APSP}.

Step 2) is easy to implement. Each node broadcasts its connectivity number, and so it takes one round for each node to create the desired list $\mathcal{L}$ of sorted connectivity numbers. Furthermore, every node can create the sets $V_j = \{i\in V: r_i \leq \rho_j\}$ since each node broadcasted its connectivity number.

Step 3) is much more involved and requires a more careful discussion. In particular, steps 1) and 2) only happen once, whereas step 3) executes the main loop for computing an approximate solution to the SNDP. 

We begin with the important definition of a \textit{shortest path forest}, or SPF, for the distributed algorithm for solving SNDP. Notably, in the distributed models, it will be easy to construct an SPF once shortest-path distances are computed. We will then first show how to solve SNDP approximately, given an SPF.
\begin{definition}{(Shortest Path Forest (SPF)):}\label{def: SPF}
For a weighted, undirected graph $G = (V, E, W)$ together with a given set of source nodes $Z = \{z_1, \dots, z_k\}$, a subgraph $F = (V, E_F, W)$ of $G$ is called a \textit{shortest path forest} if it consists of $|Z|$ disjoint trees $T_z = (V_z, E_z, W)$ satisfying
\begin{itemize}
    \item[i)]$z_i \in T_{z_j}$ if and only if $i =j$, for $i, j \in [k]$.
    \item[ii)] For each $v\in Z_i, d_G(v, z_i) = \min_{z \in Z}d_G(v,z)$, and a shortest path connecting $v$ to $z_i$ in $G$ is contained in $T_{z_i}$
    \item[iii)] The $V_{z_i}$ partition $V$, and $E_{z_1} \cup E_{z_2} \dots \cup E_{z_k}=E_F \subset E$
\end{itemize}
\end{definition}

In other words, an SPF is a forest obtained by gathering, for each node, a shortest path in $G$, connecting it to the closest source node in $Z$. For a node $v$ in a tree, let $par(v)$ denote the parent node of $v$ in that tree, $s(v)$ the source node in the tree that $v$ will be in, and $ID(v)\in [n]$ the ID of node $v \in V$. Let $\mathfrak{S}(v):= \{z: d_G(v, z) = \min_{z\in Z} d_G(v,z)\}$ be the set of source nodes closest to node $v$. 

The following gives a simple distributed algorithm for SPF in the CONGEST-CLIQUE setting: 
\begin{mdframed}\namedlabel{alg:spf}{\texttt{DistributedSPF}}
\begin{enumerate}
    \item [\namedlabel{itm:SPF1}{1}:] Each node $v$ sets $s(v) := argmin_{z\in \mathfrak{S}(v)} ID(z)$.
    \item [\namedlabel{itm:SPF2}{2}:] Each node $v$ sets $par(v):= R_v(s(v))$, $R_v$ being the routing table of $v$, and sends a message to $par(v)$ to indicate this. If $v$ receives such a message from another node $u$, it registers $u$ as its child in the SPF. 
\end{enumerate}
\end{mdframed}
Step \ref{itm:SPF1} in \ref{alg:spf} requires no communication, as does
each node $v$ setting $par(v)$ in step \ref{itm:SPF2}, if one assumes that APSP and routing tables have already been computed (which was done for our algorithm in step 1). Thus, step \ref{itm:SPF2} requires $1$ round of communication of up to $n$ (classical) messages.
\begin{claim}
After executing the \ref{alg:spf} procedure, the trees $T_{z_k} = (V_{z_k}, E_{z_k}, W)$ with $V_{z_k} := \{v\in V: s(v) = z_k\}$ and $E_{z_k} := \{v, par(v)\}: v\in V_{z_k}\}$ form an SPF.
\end{claim}
\begin{proof}
i) holds since each source node is closest to itself. iii) is immediate. For ii), notice that for $v \in V_{z_k}$, $par(v) \in V_{z_k}$ and $\{v, par(v)\}\in E_{z_k}$ as well. Then $par(par(\dots par(v) \dots)) = z_k$ and the entire path to $z_k$ lies in $T_{z_k}$. 
\end{proof}
Hence, after this procedure, we have an SPF on the graph, with each node knowing its SPF parent-child relationships. 

Once one has computed the SPF, one can obtain an analog of an MST on the complete distance graph on the source nodes of the SPF by finding the minimum weight way to connect all the source nodes of the SPF. 
We modify the edge weights before constructing an MST on that new graph to achieve this, as done in \cite{Saikia_Karmakar2019_SteinerTree_CONGESTCLIQUE}. Partition the edges $E$ into three sets -- \textit{tree edges} $E_F$ as in \ref{def: SPF} that are part of the edge set of the SPF, \textit{intra-tree edges} $E_{IT}$ incident on two nodes in the same tree $T_i$ of the SPF, and \textit{inter-tree edges} $E_{XT}$ that are incident on two nodes in different trees of the SPF.
Having each node know which of these its edges belong to can be done in one round by having each node send its neighbors the ID of the source node it chose as the root of the tree in the SPF it is a part of.
Then, the edge weights are modified as follows, denoting the modified weights as $W'$:
\begin{itemize}
    \item[(i):] For $e =(u,v)\in E_T, W'(u,v) := 0$
    \item[(ii):] For $e = (u,v) \in E_{IT}, W'(u,v) := \infty$
    \item[(iii):] For $e = (u,v)\in E_{XT}, W'(u,v) := d(u, Z_u) + W(u,v) + d(v, Z_v)$,
\end{itemize}
where $d_G(u, s(u))$ is the shortest-path distance in $G$ from $u$ to its closest source node.

Next, we find a minimum spanning tree on the graph $G' = (V, E, W')$, for which we may implement the classical $\mathcal{O}(1)$ round algorithm proposed by \cite{MST_in_O1_CCM_Nowicki2019}.
On a high level, this constant-round complexity is achieved by sparsification techniques, reducing MST instances to sparse ones, and then solving those efficiently.
We skip the details here and refer the interested reader to \cite{MST_in_O1_CCM_Nowicki2019}. 
After this step, each node knows which of its edges are part of this weight-modified MST, as well as the parent-child relationships in the tree for those edges.

Finally, we prune this MST by removing non-source leaf nodes and the corresponding edges.
This is done by each node $v$ sending the ID of its parent in the MST 
to every other node in the graph.
As a result, each node can locally compute the entire MST and then decide whether it connects two of the source nodes.
If it does, it decides it stays part of the tree; otherwise, it broadcasts that it is to be pruned.
Each node that has not been pruned then registers the edges connecting it to non-pruned neighbors as part of the tree.

This pruning step takes $2$ rounds and up to $n^2+n$ classical messages.

After this procedure, one is left with exactly the desired tree that is needed in step 3) for the general algorithmic framework we follow. Namely, consider the complete distance graph $\bar{G}$, and its subgraph restricted to the vertices $V_j$ in iteration j of step 3). The tree obtained by the above method is exactly a minimum spanning tree on this subgraph, except that rather than having any edge $uv$ of the complete distance graph $\bar{G}$, in the tree (which is on $G$, not $\bar{G}$) we have the shortest path certifying the distance from $u$ to $v$ instead. We forego a formal proof of this fact but give some valuable intuition on how this happens. The SPF contains one tree for each of its root nodes, which are the nodes in $V_j$. Every other node in $G$ is a part of one of these trees, such that the path in the tree it is a member of to the root node of that tree is the shortest path from the node in question to the root node of the tree. Next, the edge weight modifications and finding an MST on the resulting graph effectively joins these trees in a minimum-weight manner. The edges in the tree are all part of shortest paths to the root nodes, and so this resulting tree joins the root nodes by the shortest paths connecting them. These shortest paths correspond to edges in $\bar{G}$, and the resulting tree is the realization of an MST over the subgraph of the root nodes $V_j$ in $\bar{G}$. We direct the reader to \cite{Saikia_Karmakar2019_SteinerTree_CONGESTCLIQUE} and \cite{Kou1981_fast_algo_for_ST_} for more detailed treatment.   

Once this tree is obtained, $\rho_j - \rho_{j-1}$ copies of every edge in the tree are added to the reported solution of the SNDP. This means that $\rho_j - \rho_{j-1}$ edge-disjoint paths between any two nodes in $V_j$ are added to the SNDP solution at iteration $j$. Noting that this tree corresponds to the desired tree for step 3 of \ref{alg: tree heuristic}, our method indeed serves as a distributed implementation thereof. In particular, it thus inherits the same guarantees as the centralized algorithm in \cite{GoemansBertsimas1993SurvivableNetworks} that we have summarized in Table \ref{table: sndp guarantees}.

\subsection{Distributed APSP and Routing Tables} \label{sec: distributed APSP}

Next, we discuss the APSP and routing table computations for step 1). First, we note that in the algorithmic framework we use, it is necessary to compute routing tables along with the shortest path distances. Although we work with the complete distance graph $\bar{G}$ following \cite{GoemansBertsimas1993SurvivableNetworks}, edges between nodes in $\bar{G}$ correspond to shortest paths between nodes in $G$. In the end, we need to report those paths in our solutions, not just edges in $\bar{G}$, so path information is needed. Exactly this necessary path information is stored in the routing tables (in \cite{GoemansBertsimas1993SurvivableNetworks}, this transformation going from edges in $\bar{G}$ to shortest paths in $G$ is discussed in Theorem 3).  We here mention the two asymptotically fastest existing algorithms in the classical and quantum CONGEST-CLIQUE models for computing APSP distances and routing tables. 

\begin{theorem}[APSP and Routing tables in cCCM \cite{CensorHillel2016_AlgebraicMethodsCCM_fast_APSP}]
    APSP and routing tables certifying the shortest path distances can be computed in $\tilde{O}(n^{1/3})$ rounds in the cCCM. 
\end{theorem}

\begin{theorem}[APSP in qCCM \cite{IzumiLeGallMag2019_APSP_QuantumDist}]
    APSP distances can be computed in $\tilde{O}(n^{1/4})$ rounds in the qCCM. 
\end{theorem}
This result is extended to computing routing tables in $\tilde{O}(n^{1/4})$ as well in \cite{Kerger2023_mind_the_O-tilde}. 

We remark that approximation algorithms for APSP exist as well in these models. However, they tend to not compute routing tables, but only APSP \textit{distances}; see for example \cite{Dory2021Constant_round_approx_APSP} or \cite{nanongkai2014distributed}. The latter does mention routing tables, but they are obtained in $O(n)$ time after computing the APSP distances, which is not useful in the CONGEST-CLIQUE model since any problem can be solved in $O(n)$ time in the CCM (this $O(n)$ routing table calculation is meaningful in the CONGEST model though, which the work focused on). To make use of faster approximation algorithms for frameworks like ours, these algorithms would need to also produce routing tables. 

We provide here some explanation on how these algorithms work, and how the quantum speedup for the APSP step is achieved. 

Both the classical and quantum algorithms for APSP use so called \textit{distance products} for the shortest path computations, defined as follows: 
\begin{definition}\label{def: distance products}
The \textit{distance product} between two $n\cross n$ matrices $A$ and $B$
is defined as the $n\cross n$ matrix $A \star B$ with entries:  
\begin{align}
    (A \star B)_{ij} = \min_k\{A_{ik} + B_{kj}\}.
\end{align}
\end{definition}
The distance product is also sometimes called the min-plus or tropical product. In particular, consider a graph $G$ with adjacency matrix $A$, whose $ij$ entry contains either the weight of an edge between node labeled $i$ and the node labeled $j$, or $\infty$ if there is no edge between nodes $i$ and $j$. The entries of the distance product $A' := A\star A$ tell us the shortest path from node $i$ to node $j$ that uses $2$ edges (or, if no such path exists, the entry is $\infty$). Then, taking $A'' := A' \star A'$ has in the $ij$ entry the shortest-path distance from node $i$ to $j$ that uses up to $4$ edges, taking $A''' := A'' \star A''$ gives the shortest path distance using up to $8$ edges, and so on. Since any shortest path uses at most $n$ edges, computing $\ceil{log_2(n)}$ distance products suffices to compute shortest path distances, and a relatively simple modification (see e.g. \cite{witness_ref_Zwick2000}, \cite{Kerger2023_mind_the_O-tilde}) allows to compute routing tables via computing $polylog(n)$ distance products as well. 
These distance products can be computed via searching for negative triangles in a graph. 

Given a weighted graph, a negative triangle is a triple of edges $\Delta^- = (uv, vw, wu) \subset E^3$ such that $\sum_{e\in \Delta^-} W_e < 0$, i.e. a $3$-clique such that the sum of the edges' weights is negative. 
Suppose that all $n$ nodes in the network can determine which of their incident edges are part of a negative triangles in $T(n)$ rounds (simultaneously). Then for two $n\cross n$ matrices $A$ and $B$ with entries in $[M]$, the distance product $A \star B$ can be computed distributedly in $T(3n)\cdot \ceil{\log_2(2M)}$ rounds.
 To make this clear, let $A$ and $B$ be arbitrary $n \cross n$ integer-valued matrices, and $D$ be an $n \cross n$ matrix initialized to $\0$.S Suppose the information is distributed across the network as the node labeled $j$ knowing the $j^{th}$ row of $A$ and $B$. We describe now how the distance product is computed through finding such negative triangles. Let each node $u$ in the network simulate three copies of itself,$u_1, u_2, u_3$, and write $V_1, V_2, V_3$ as these node-copy sets that are being simulated by all the nodes in $V$. Take $G' = (V_1 \cup V_2 \cup V_3, E', W')$, letting $u_iv_j \in E'$ for $u_i \in V_i, v_j \in V_j, i\neq j$, taking $W_{u_1v_2}' = A_{uv}$ for $u_1 \in V_1, v_{2} \in V_{2}$, 
  $W_{u_2v_3}' = B_{uv}$ for $u_2 \in V_2, v_{3} \in V_{3}$,
  and $W_{u_3v_1}' = D_{uv}$ for $u_3 \in V_3, v_1 \in V_{1}$.
Notice that the following condition is equivalent to an edge $zv$ belonging to a negative triangle in $G'$:
$$
\min_{u \in V} \{A_{vu} + B_{uz}\} < -D_{zv}.
$$
Assuming we can find negative triangles for a $k$-node graph in $T(n)$ rounds, 
with a non-positive matrix $D = \mathbf{0}$ initialized we can apply simultaneous binary searches on $D_{zv}$, with values between $\{-2M, 0\}$, updating it for each node $v$ after each run of triangle finding to find 
$\min_{u \in V} \{A_{vu} + B_{uz}\}$ for every other node $z$ in\\ $T(3n)\cdot \ceil{\log(\max_{v, z \in V}\{\min_{u \in V}\{ A_{vu}+B_{uz}\}\})}$ rounds, since $G'$ is a tripartite graph with $3n$ nodes.
With this reduction to triangle finding established, we provide a high-level idea of how the triangle finding can be done faster in the quantum model than in the classical model, and refer the reader to \cite{IzumiLeGallMag2019_APSP_QuantumDist} and \cite{Kerger2023_mind_the_O-tilde} for details in the quantum model, and \cite{CensorHillel2016_AlgebraicMethodsCCM_fast_APSP} and \cite{Dolev_lenzen_2012TriTA} for a rigorous treatment in the classical model. Since triangles consist of three nodes, a triangle in $G$ is an element of the space $\bV^3 := V \times V \times V$. Consider partitioning that space, and to each node in the network assign the task of searching through one piece of the partition for negative triangles. Take sets of the form $\bV_1 \times \bV_2 \times \bV_3$ to partition $\bV^3$, where elements of the $\bV_i$ are sets of nodes. In the quantum case, one can take $|\bV_1| = |\bV_2| = n^{1/4}$ and $|\bV_3| = n^{1/2}$, with $\w\in \bV_1$ or $\bV_2$ to contain $n^{3/4}$ nodes, and $\w \in \bV_3$ to contain $n^{1/2}$ nodes. That is to say, the space $\bV^3$ is partitioned by twice taking $n^{1/4}$ subsets of $n^{3/4}$ nodes each, and once taking $n^{1/2}$ subsets of $n^{1/2}$ nodes each. In this way, $\bV_1 \times \bV_2 \times \bV_3$ has $n$ elements (that are each sets of nodes), and so one can assign one node in the network the task of finding triangles in one of these subsets to cover all of them. Through efficient routing schemes, it is possible to send the information about the relevant elements of $\bv_1$ and $\bV_2$ to the relevant node in $\tilde{O}(n^{1/4})$ rounds, and then the node will search one of the sets of nodes in $\bV_3$ for triangles formed with the nodes in $\bV_1, \bV_2$. Since that will be a set of $n^{1/2}$ nodes, one can use a distributed Grover search to find triangles searching through this set in $\tilde{O}(\sqrt{n^{1/2}}) = \tilde{O}(n^{1/4})$ time. 
In the classical case, the partitioning needs to be done differently. Namely, searching through the above $\bV_3$ for triangles would take $n^{1/2}$  time classically. Instead, one can create a different partition so that the sets have $n^{1/3}$ elements each in them, so that both the routing and searching take $n^{1/3}$ time each. The details of such a $O(n^{1/3})$ procedure for triangle finding in the classical CONGEST-CLIQUE can be found in \cite{Dolev_lenzen_2012TriTA}. So in effect, the source of the quantum speedup for the triangle finding problem is that the approach of partitioning $V\times V \times V$ and then routing information about one partition set to one node each, ends up making the node in question implement an unstructured search. Since distributed Grover search gives a quadratic speedup for the unstructured search, one is able to speed up the entire process, but the speedup is a bit less than quadratic due to the balancing of the partitioning scheme involved.

\section{On Potential Improvements}
Finally, we discuss potential improvements to the algorithm, first by noting possible local improvements that can be made as proposed in the final step 4) of the algorithm. After each tree to be added to the solution the algorithm reports is computed in the iterations of step 3, that tree can be broadcasted to the entire network in a single round. Namely, each node in the tree can broadcast its parent in the tree and the associated weight. Hence, we may assume that at the end of step 3), each node has information about the full approximate solution at the time. Notice that through the method of taking a tree and taking multiple copies of edges in it, in many cases one may be able to improve on the approximate solution by adding an edge to create a cycle and then deleting a copy of each of the other edges along the cycle. For example, in the simple case of a triangle graph with connectivity requirement 2 for each node and weight 1 on each edge, the tree heuristic gives the solution of using two edges twice. Adding the third edge creates a cycle, so that if one deletes the second copies of the other two edges, there are still two paths from each node to every other node (one by going along the cycle in each direction). Hence the connectivity requirements remain satisfied, but the updated solution has lower weight. Since each node has full information about $\hat{S}$, it can locally determine whether any edges it is incident upon can create such a cycle that can lead to an improvement, and broadcast that edge to the network. Then every node receives a list of edges for proposed improvements and can compute the updated best SNDP solution given the information about those edges.

In fact, a similar idea of starting with a tree and then finding ``shortcuts" is used in Christofides' approximation algorithm for the TSP \cite{Christofides1976WorstCaseAO}, which is generalized to the \textit{improved tree heuristic} in \cite{GoemansBertsimas1993SurvivableNetworks}. This involves computing minimum weight perfect matchings after each iteration of step 3 of \ref{alg: tree heuristic} to find improvements to the simpler heuristic of adding multiple copies of each edge of the tree to the solution subgraph. With this improvement, the algorithm gives provably better solutions in some cases of the SNDP, and so the distributed algorithms here could be improved via fast distributed minimum weight perfect matching (MWPM) algorithms. We note that the MWPM problem is also of high importance in the field of quantum error correction, and so such algorithms would be of interest in their own right.

\section*{Acknowledgements} 
We are grateful for support from the NASA Ames Research Center, from the NASA SCaN
program, and from DARPA under IAA 8839, Annex 130. PK and DB acknowledge support from the
NASA Academic Mission Services (contract NNA16BD14C). 



{\small
 \bibliographystyle{apacite} 
 \bibliography{quantumcite}
}

\end{document}